\def\year{2017}
\documentclass[letterpaper]{article}
\usepackage{aaai17}
\usepackage{times}
\usepackage{helvet}
\usepackage{courier}
\usepackage{graphicx}
\usepackage{amsmath}
\usepackage{amsfonts}
\usepackage{amsthm}
\usepackage{algorithm}
\usepackage{url}
\usepackage{fixltx2e}
\usepackage{adjustbox}
\usepackage[noend]{algpseudocode}
\usepackage{caption}
\usepackage{subcaption}
\usepackage{array}
\usepackage[utf8]{inputenc}
\begin{document}

\theoremstyle{plain}
\newtheorem{thm}{Theorem}
\newtheorem{lem}[thm]{Lemma}
\newtheorem{prop}[thm]{Proposition}
\newtheorem{corr}{Corollary}
\theoremstyle{definition}
\newtheorem{defn}{Definition}
\newtheorem{conj}{Conjecture}
\newtheorem{exmp}{Example}
\theoremstyle{remark}
\newtheorem*{rem}{Remark}
\newtheorem*{note}{Note}
\newtheorem{case}{Case}
\algnewcommand\algorithmicforeach{\textbf{for each}}
\algdef{S}[FOR]{ForEach}[1]{\algorithmicforeach\ #1\ \algorithmicdo}

\title{Edge N-Level Sparse Visibility Graphs: Fast Optimal Any-Angle Pathfinding Using Hierarchical Taut Paths}
\author{
    Shunhao Oh \and Hon Wai Leong\\
    Department of Computer Science\\
    National University of Singapore\\
    ohoh@u.nus.edu, leonghw@comp.nus.edu.sg
}
\maketitle
\begin{abstract}
In the Any-Angle Pathfinding problem, the goal is to find the shortest path between a pair of vertices on a uniform square grid, that is not constrained to any fixed number of possible directions over the grid. Visibility Graphs are a known optimal algorithm for solving the problem with the use of pre-processing. However, Visibility Graphs are known to perform poorly in terms of running time, especially on large, complex maps. In this paper, we introduce two improvements over the Visibility Graph Algorithm to compute optimal paths. Sparse Visibility Graphs (SVGs) are constructed by pruning unnecessary edges from the original Visibility Graph. Edge N-Level Sparse Visibility Graphs (ENLSVGs) is a hierarchical SVG built by iteratively pruning non-taut paths. We also introduce Line-of-Sight Scans, a faster algorithm for building Visibility Graphs over a grid. SVGs run much faster than Visibility Graphs by reducing the average vertex degree. ENLSVGs, a hierarchical algorithm, improves this further, especially on larger maps. On large maps, with the use of pre-processing, these algorithms are orders of magnitude faster than existing algorithms like Visibility Graphs and Theta*.
\end{abstract}

\section{Introduction}
In many pathfinding applications involving open spaces, it is common strategy to abstract a 2D map into a uniform square grid \cite{comprehensive_study}. Many grid-based pathfinding algorithms are 8-directional, where the agent can only move in the four cardinal and four diagonal directions along the grid. We consider the Any-Angle Pathfinding problem, where this constraint is removed. The start and goal are vertices of the grid, The objective is to compute shortest path in terms of euclidean length, from the start to the goal, that does not intersect any blocked tiles in the grid.

There are many optimal algorithms for 8-directional pathfinding, like a simple 8-directional A*, or faster algorithms like Jump-Point Search \cite{harabor_jps} and Subgoal Graphs \cite{subgoal_graphs}. On the other hand, computing optimal any-angle paths is more difficult. Thus, many existing Any-Angle Pathfinding algorithms like Theta* \cite{nash_thetastar} and Block A* \cite{block_astar}, are heuristic in nature.

A known optimal Any-Angle Pathfinding algorithm is A* on Visibility Graphs \cite{visibilitygraphs}. However, Visibility Graphs can be inefficient in practice for two reasons. Firstly, Visibility Graph construction requires many Line-of-Sight Checks, quadratic on number of tiles in the grid. While this can be partially solved by pre-processing the visibility graph, a second issue is the high average vertex degree, slowing down an A* search on the graph.

Another algorithm, Anya \cite{harabor_anya} has been shown to compute optimal paths efficiently, with comparable speeds to heuristic algorithms like Theta* \cite{harabor_anya_full}. It also has the advantage of being an online algorithm, requiring no pre-processing of the map.

In this paper, we introduce two improvements to the Visibility Graph algorithm, Sparse Visibility Graphs (SVGs) and Edge N-Level Sparse Visibility Graphs (ENLSVGs), which are orders of magnitude faster than existing algorithms Theta* and Visibility Graphs. The relationship to other algorithms can be found in \cite{comparison_paper}. SVGs are constructed from removing unnecessary edges from the Visibility Graph. ENLSVGs are constructed by building a hierarchy over an underlying SVG.

The SVG and ENLSVG algorithms are fast and optimal, but are offline algorithms, using a slower pre-computation step so that many shortest path queries can be made quickly. A drawback of offline algorithms is that the pre-computation step needs to be repeated each time the map changes.

Both algorithms, SVGs and ENLSVGs, are based only on the simple concept of pruning non-taut paths to reduce the search space. Through these algorithms, we observe the relationship between taut and optimal paths. Optimal paths are difficult to compute in general, but taut paths, being locally optimal rather than globally optimal, can be computed very easily in constant time. We show how just simple taut path restrictions can greatly reduce the search space for an optimal search. Pruning taut paths on one level forms SVGs, and extending it to n levels of pruning forms ENLSVGs.

Previous work making use of taut paths in Any-Angle search include Anya \cite{harabor_anya} and Strict Theta* \cite{oh_stricttheta}. The idea of building a multi-level hierarchy for optimal pathfinding is based on previous work on N-Level Subgoal Graphs \cite{aa_subgoal_graphs}. N-Level Subgoal Graphs prune vertices using shortest paths, while ENLSVGs prune edges using taut paths.

We also introduce Line-of-Sight Scans, a fast algorithm for querying visible neighbours of a vertex in the grid. This replaces Line-of-Sight Checks for building the Visibility Graph and inserting the start and goal points into the graph.



\section{Preliminaries}

As previously mentioned, A* on Visibility Graphs (VGs) returns optimal any-angle paths. The vertices of a Visibility Graph consists of the start and goal vertices, and all convex corners of obstacles. We connect all pairs of vertices with Line-of-Sight. Visibility Graph construction is a slow process as it requires Line-of-Sight Checks between every pair of vertices. Thus, it is more reasonable to pre-process a Visibility Graph on the grid, and reuse the graph for multiple shortest-path queries.

As start and goal vertices differ for each query, we leave them out of the pre-processed graph. During a shortest-path query, we connect the start and goal vertices to the Visibility Graph by doing Line-of-Sight Checks to all other existing vertices. We remove them after the query.

We use taut path restrictions to reduce the search space. Informally, a taut path is a path which, when treated as a string, cannot be made ``tighter" by pulling on its ends. Formally and practically, a path is taut if and only if every heading change in the path wraps tightly around some obstacle \cite{oh_stricttheta}. As shown in Figure \ref{fig:tautpaths}, only a single obstacle needs to be checked per heading change to determine if a path is taut (Figure \ref{fig:tautpaths3} can never be taut). As all optimal paths are taut, if we restrict the search space to taut paths, the optimal path will be included in the search space.

\begin{figure}[!h]
  \centering
    \begin{subfigure}[b]{.33\linewidth}
      \centering
      \includegraphics[width=.8\linewidth]{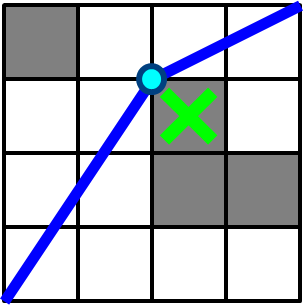}
      \caption{}
      \label{fig:tautpaths1}
    \end{subfigure}%
    \begin{subfigure}[b]{.33\linewidth}
      \centering
      \includegraphics[width=.8\linewidth]{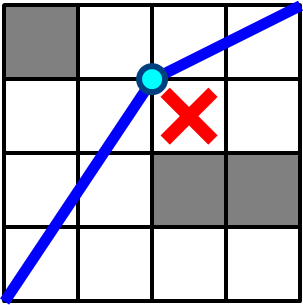}
      \caption{}
      \label{fig:tautpaths2}
    \end{subfigure}%
    \begin{subfigure}[b]{.33\linewidth}
      \centering
      \includegraphics[width=.8\linewidth]{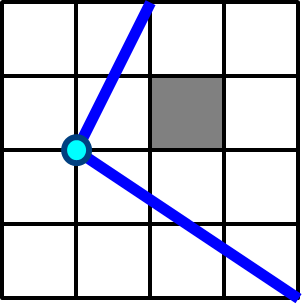}
      \caption{}
      \label{fig:tautpaths3}
    \end{subfigure}%
  \caption{(a) is taut, while (b) and (c) are not.}
  \label{fig:tautpaths}
\end{figure}

For the rest of this paper, we make use of Taut A* for graph search in place of the standard A* algorithm. In Taut A*, whenever we attempt to generate a successor $v$ from the current state $u$, we first check for tautness. $v$ can be a successor of $u$ only if the subpath $parent(u)-u-v$ is taut.

\section{Sparse Visibility Graphs}

Many edges in the Visibility Graph are unnecessary, as they are never used in the final path found by A*. In particular, these edges cannot be part of any taut path between any pair of start or goal points, unless the start or goal is one of the edge's endpoints, in which case the edge will be added anyway when connecting the start or goal to the visibility graph.

Refer to the edge $uv$ in Figure \ref{fig:onetautexit}. Suppose endpoint $v$ is neither the start nor the goal. Thus, in any path involving edge $uv$, the path must leave $v$ to move to another vertex. However, in all legal directions the path can leave $v$ from, the path will not be taut, and thus not optimal. We say that the edge $uv$ has no taut exit from $v$.

\begin{figure}[!h]
  \centering
    \begin{subfigure}[b]{.5\linewidth}
      \centering
      \includegraphics[width=.9\linewidth]{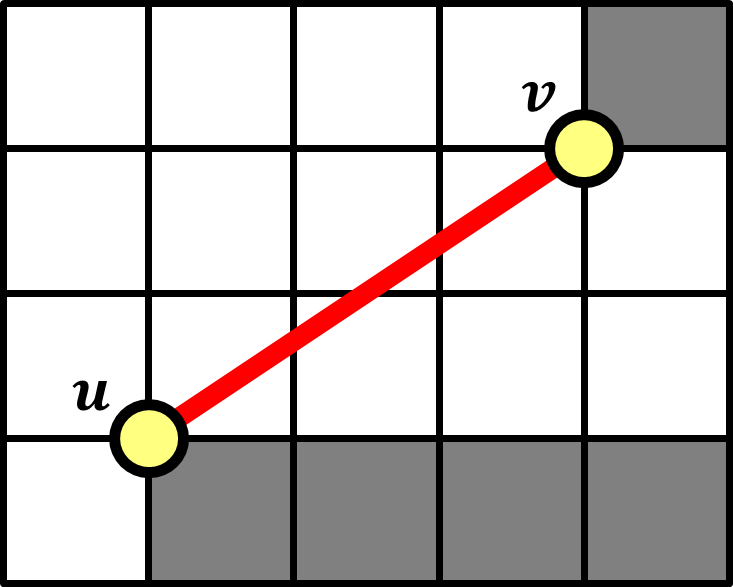}
      \caption{Edge with no taut exit from $v$}
      \label{fig:onetautexit}
    \end{subfigure}%
    \begin{subfigure}[b]{.5\linewidth}
      \centering
      \includegraphics[width=.9\linewidth]{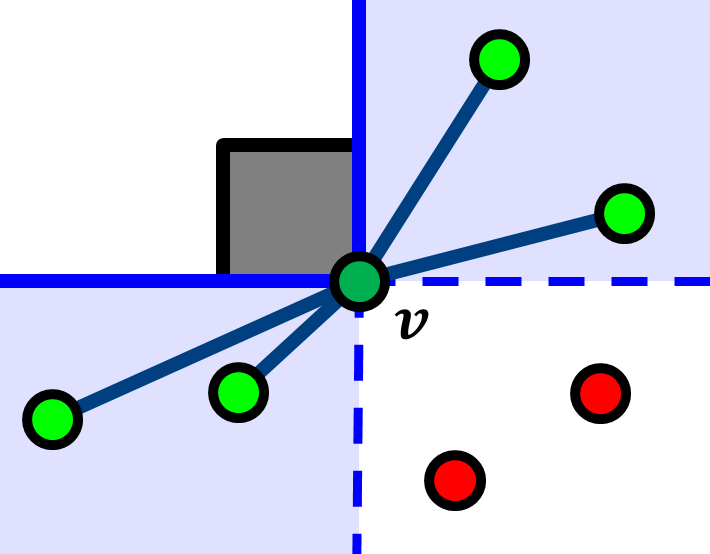}
      \caption{Taut regions (blue) around $v$}
      \label{fig:tautregions}
    \end{subfigure}%
  \caption{Edges without taut exit directions are unnecessary.}
  \label{fig:svg_tautexits}
\end{figure}

To identify the edges to be pruned, we consider the taut regions around each vertex $v$ in the graph. A vertex $u$ is in the taut region of vertex $v$ if the edge $uv$ has a taut exit from $v$. To find the taut regions, we need only consider the obstacles adjacent to $v$ as shown in Figure \ref{fig:tautregions}. We prune any edge $uv$ where any one of the endpoints does not lie within the taut region of the other endpoint (Figure \ref{fig:onlyoneendpointtaut}). The remaining edges make up the Sparse Visibility Graph (SVG).

\begin{figure}[!h]
  \centering
  \includegraphics[width=.65\linewidth]{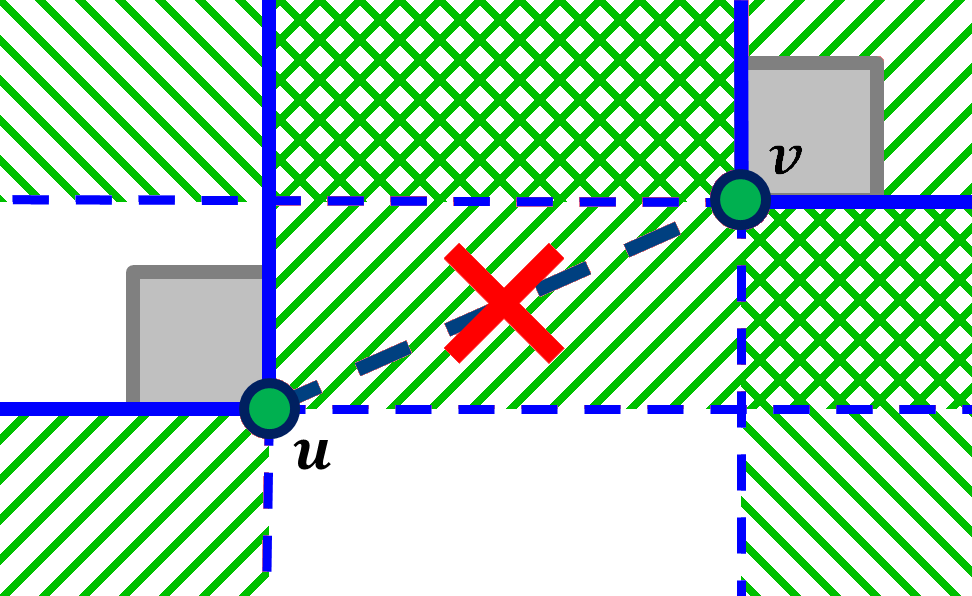}
  \caption{$v$ lies within the taut region of $u$, but $u$ does not lie within the taut region of $v$. Edge $uv$ is pruned.}
  \label{fig:onlyoneendpointtaut}
\end{figure}

\subsection{Collinear Points}

We describe our policy on collinear points in an SVG. Naively, a set of $k$ collinear points would form a size $k$ clique due to Line-of-Sight between any two points in the set (Figure \ref{fig:collinearpoints}). This is clearly wasteful and unnecessarily increases the average vertex degree. In these cases, it suffices for each vertex to have edges only to its closest neighbour on each side of the point on the line. Intuitively, we can imagine each vertex as being an epsilon-size Line-of-Sight obstruction.

\begin{figure}[!h]
  \centering
    \includegraphics[width=.65\linewidth]{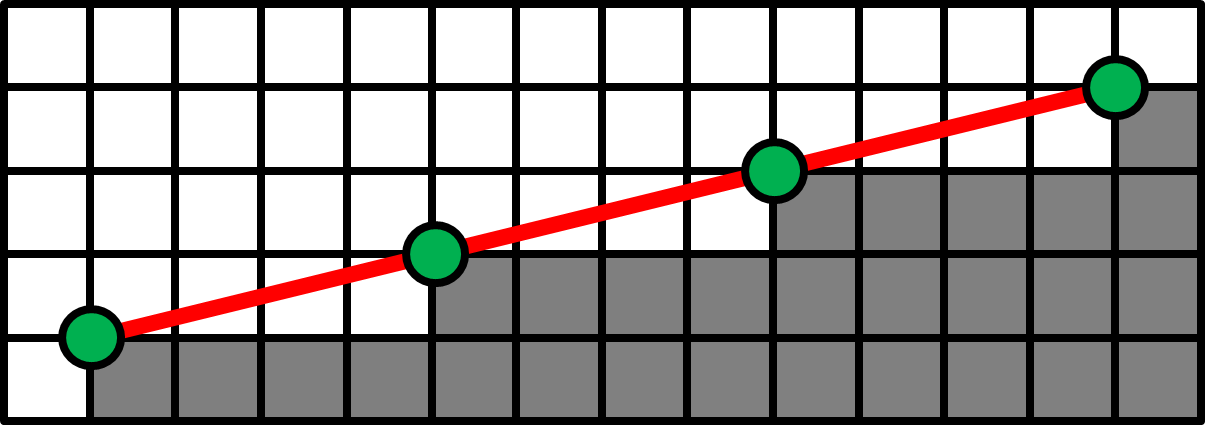}
  \caption{Every pair of points on the line has Line-of-Sight.}
  \label{fig:collinearpoints}
\end{figure}

\subsection{Fast Construction Using Line-of-Sight Scans}

Constructing Visibility Graphs using Line-of-Sight check between every pair of vertices takes $\Theta(V^2)$ Line-of-Sight Checks even in the best case. This is because the computation is non-local. Even on dense maps where Line-of-Sight is uncommon, Line-of-Sight Checks are still conducted between vertices on opposite ends of the map.

In place of individual vertex-to-vertex Line-of-Sight Checks, we introduce Line-of-Sight Scans, which computes the set of visible vertices from a single vertex. Intuitively, a Line-of-Sight Scan from a vertex $v$ is a radial outwards scan which breaks whenever it hits an obstacle. We implement this using a similar method to the interval search used by Anya \cite{harabor_anya}.

The key advantage of Line-of-Sight Scans is that it is local. For each vertex, the running time of Line-of-Sight Scans depends on the number of visible vertices, while Line-of-Sight Checks depends on the total number of vertices in the entire map.  Line-of-Sight Scans is much faster, especially on larger maps with a low likelihood of cross-map visibility.

\begin{figure}[!h]
  \centering
    \begin{subfigure}[b]{.5\linewidth}
      \centering
      \frame{\includegraphics[width=.95\linewidth]{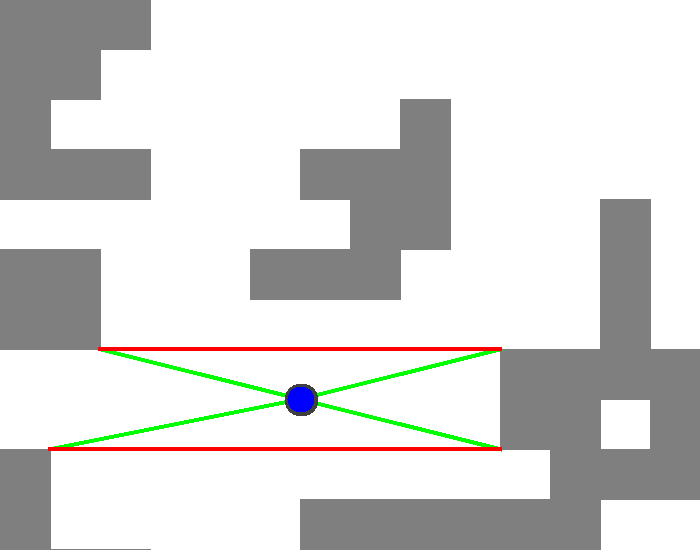}}
      \caption{Initial intervals}
      \label{fig:alldirlos_initial}
    \end{subfigure}%
    \begin{subfigure}[b]{.5\linewidth}
      \centering
      \frame{\includegraphics[width=.95\linewidth]{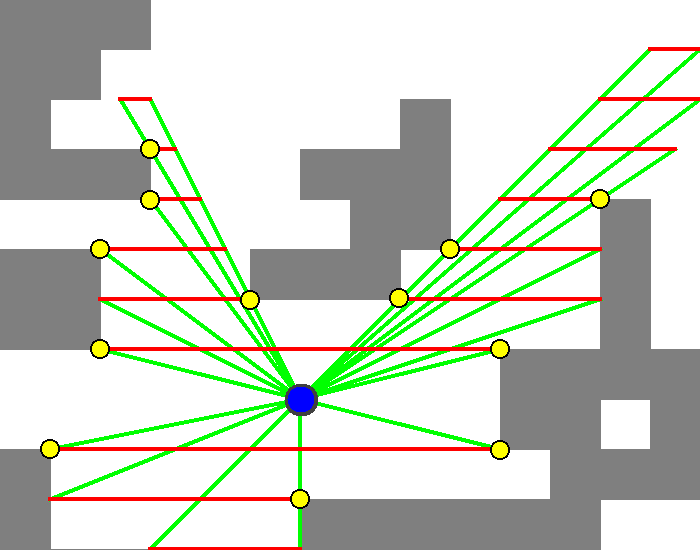}}
      \caption{Resulting search tree}
      \label{fig:alldirlos_tree}
    \end{subfigure}%
  \caption{An All-Direction Line-of-Sight scan. Intervals are the red horizontal lines. The found visible successors are marked in yellow.}
  \label{fig:alldirlos}
\end{figure}

We initialise the scan around a point (the source) by generating horizontal intervals around it as shown in Figure \ref{fig:alldirlos_initial}. Each interval is a tuple $(x_L,x_R,y)$ consisting of an integer $y$-coordinate and two fractional endpoints on the $x$-axis. The successors of an interval are the observable successors defined in \cite{harabor_anya}, which are computed by projecting the current interval onto the next $y$-coordinate away from the source. Obstacles split up generated intervals.

\begin{figure}[!h]
  \centering
  \includegraphics[width=.6\linewidth]{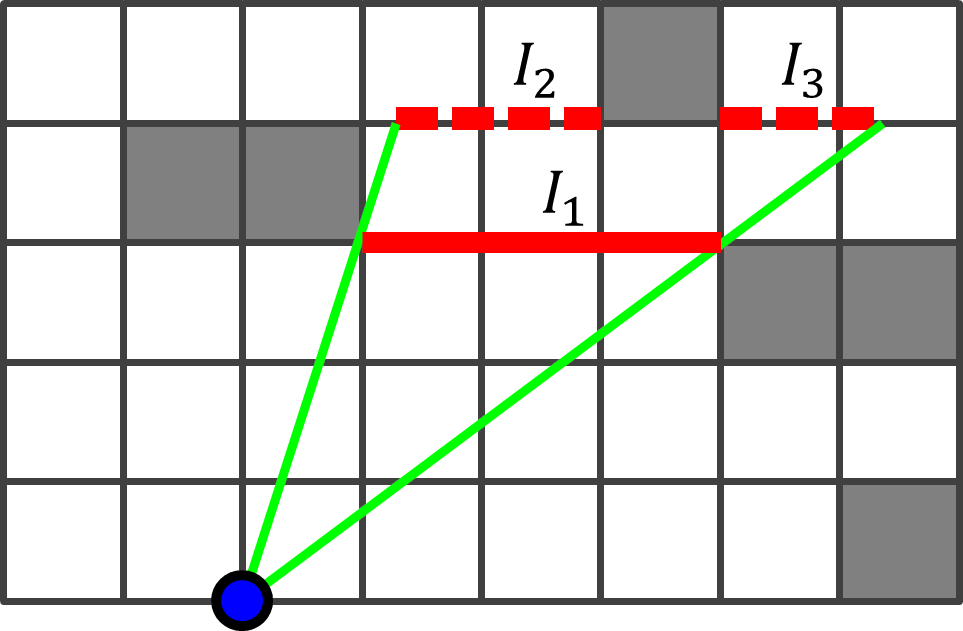}
  \caption{Successor intervals $I_2$ and $I_3$ generated from $I_1$.}
  \label{fig:observable}
\end{figure}

From there, we conduct a depth-first search over the intervals, with each interval generating its observable successors, forming the search tree in Figure \ref{fig:alldirlos_tree}. As visibility graph vertices only occur at the endpoints of the intervals, it suffices to check the interval endpoints to obtain the list of visible successors. We call this an All-Direction Line-of-Sight Scan.

In a Sparse Visibility Graph, we only add edges to vertices in taut regions (Figure \ref{fig:onetautexit}), which are determined by the current vertex's adjacent obstacles. Figure \ref{fig:sixtautcases} illustrates the six different cases. Thus, for each vertex we need only scan within the taut regions. We do this by simply changing the initial states of the search as shown in Figure \ref{fig:tautdirlos}. We call this a Taut-Direction Line-of-Sight Scan.

\begin{figure}[!h]
  \centering
    \includegraphics[width=\linewidth]{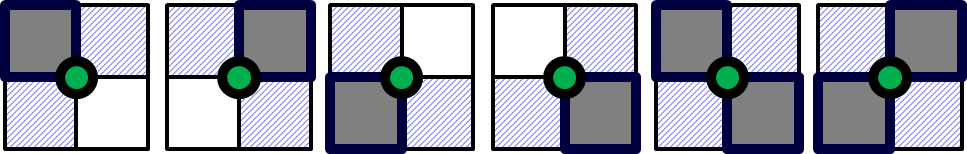}
  \caption{The six different obstacle (grey) configurations that determine taut regions (in blue).}
  \label{fig:sixtautcases}
\end{figure}

\begin{figure}[!h]
  \centering
    \begin{subfigure}[b]{.5\linewidth}
      \centering
      \frame{\includegraphics[width=.95\linewidth]{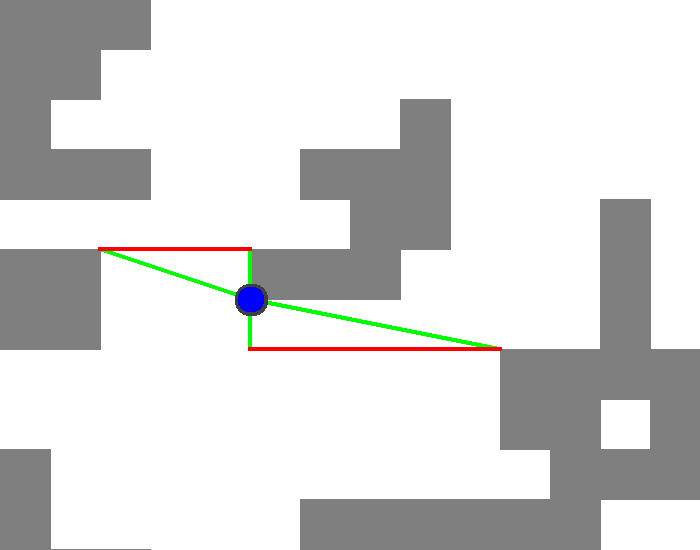}}
      \caption{Initial intervals}
      \label{fig:tautdirlos_initial}
    \end{subfigure}%
    \begin{subfigure}[b]{.5\linewidth}
      \centering
      \frame{\includegraphics[width=.95\linewidth]{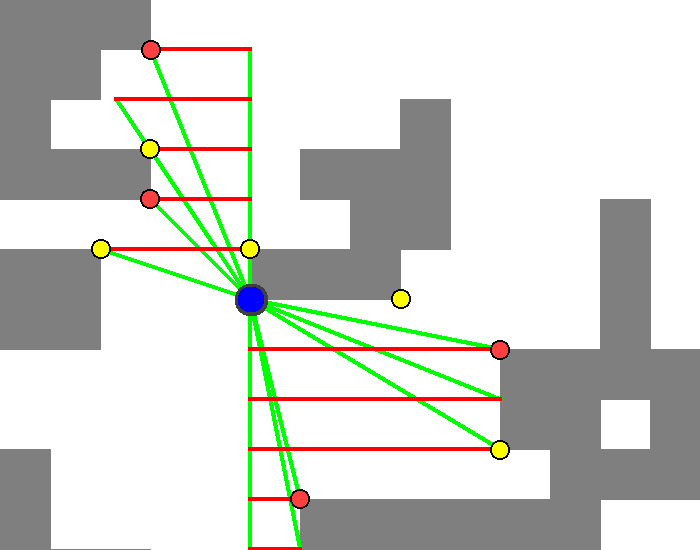}}
      \caption{Resulting search tree}
      \label{fig:tautdirlos_tree}
    \end{subfigure}%
  \caption{A Taut-Direction Line-of-Sight scan. The found visible neighbours coloured red in (b) are also pruned as they do not meet the condition shown in Figure \ref{fig:onlyoneendpointtaut}.}
  \label{fig:tautdirlos}
\end{figure}

Line-of-Sight Scans can be sped up by pre-computing left and right extents for each grid vertex - the number of tiles one can traverse in that direction before hitting an obstacle, as shown in implementation of Anya in \cite{comparison_paper}. This pre-computation also improves the runtime speed of the algorithm, as All-Direction Line-of-Sight Scans are used to insert start and end points into the graph.

\subsection{Properties of Sparse Visibility Graphs}

The Sparse Visibility Graph Algorithm simply uses Taut A* over a pre-processed SVG. SVGs reduce the average vertex degree with no cost to optimality. On randomly generated maps with percentages of blocked tiles ranging between $6\%$ and $40\%$, the average vertex degree of VGs remains approximately $2.5$ times that of SVGs. We see this in Figure \ref{fig:averagedegreechart}.

\begin{figure}[!h]
  \centering
    \includegraphics[width=.8\linewidth]{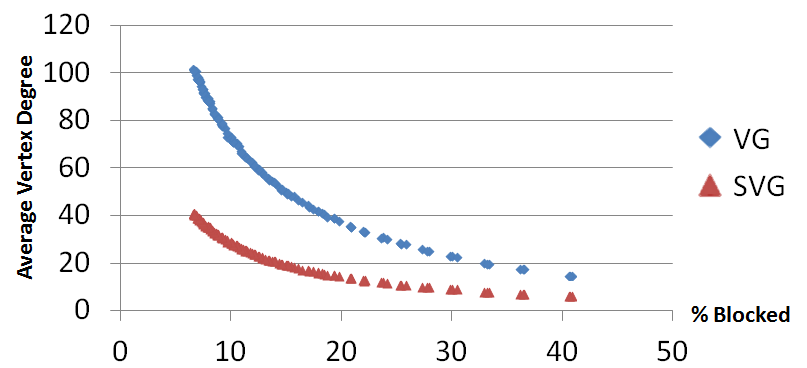}
  \caption{Comparison of average vertex degree on randomly-generated maps of various blocked densities.}
  \label{fig:averagedegreechart}
\end{figure}

Also, from Figure \ref{fig:ebonlakes_searchtrees1}, we can see that the search tree of the Sparse Visibility Graph algorithm is more sparse than that of the original Visibility Graph algorithm.

\begin{figure}[!h]
  \centering
    \begin{subfigure}[b]{.5\linewidth}
      \centering
      \includegraphics[width=.9\linewidth]{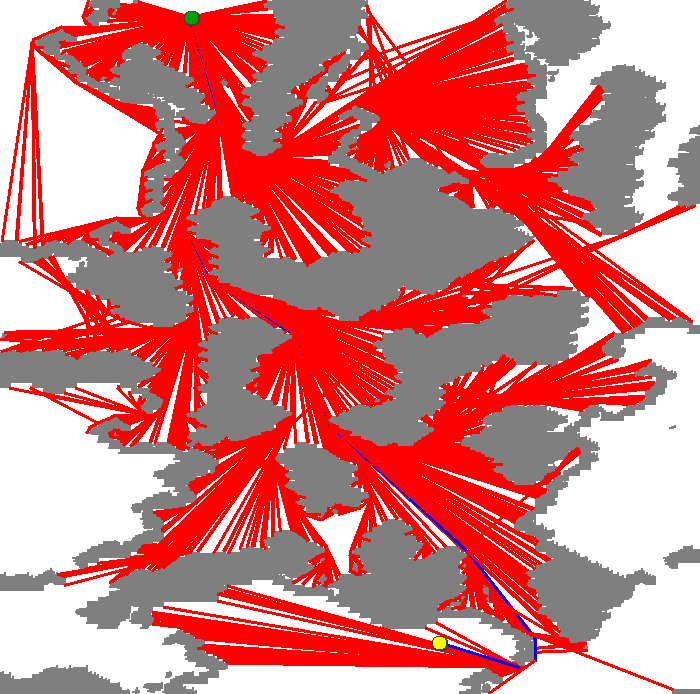}
      \caption{VG}
      \label{fig:ebonlakes_vg}
    \end{subfigure}%
    \begin{subfigure}[b]{.5\linewidth}
      \centering
      \includegraphics[width=.9\linewidth]{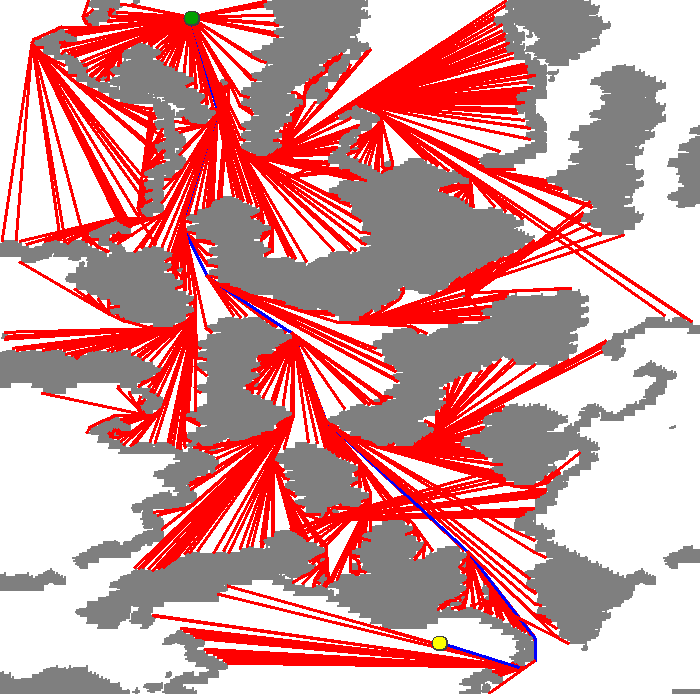}
      \caption{SVG}
      \label{fig:ebonlakes_svg}
    \end{subfigure}
  \caption{Search tree comparson on the map EbonLakes.}
  \label{fig:ebonlakes_searchtrees1}
\end{figure}

A key property of SVGs is that it almost cannot be pruned any further. Theorem \ref{thm:necessary_edges} describes this property:

\begin{thm}
\label{thm:necessary_edges}
For each edge in the Sparse Visibility Graph, there exist two points which has an optimal path that uses that edge, neither of which are the endpoints of the edge.
\end{thm}

\begin{proof}
Edges in the SVG each belong to one of the four cases in Figure \ref{fig:necessary_edges}. In each case, the two required points are marked with blue crosses. The tiles that are necessarily unblocked in each case are marked with dotted lines.
\end{proof}

\begin{figure}[!h]
  \centering
    \begin{subfigure}[b]{.35\linewidth}
      \centering
      \includegraphics[width=.9\linewidth]{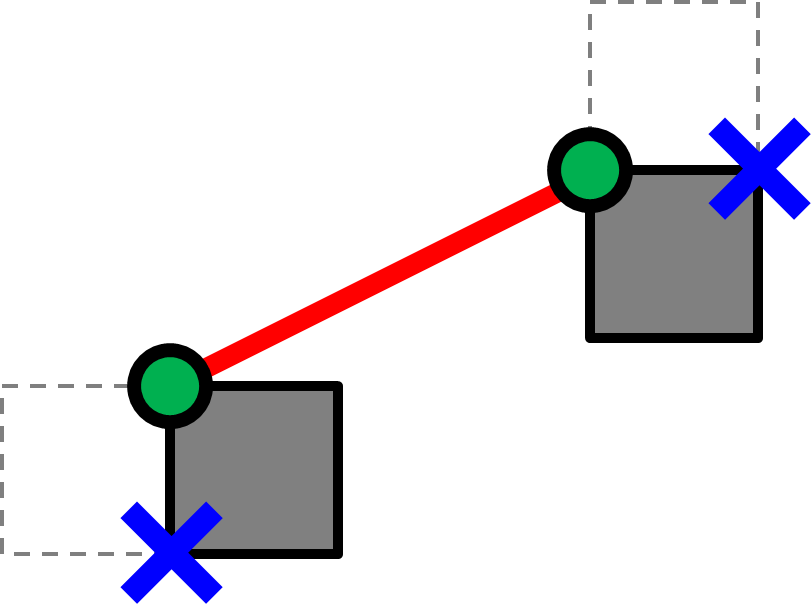}
      \caption{}
      \label{fig:necessary_c1}
    \end{subfigure}~~~~~%
    \begin{subfigure}[b]{.35\linewidth}
      \centering
      \includegraphics[width=.9\linewidth]{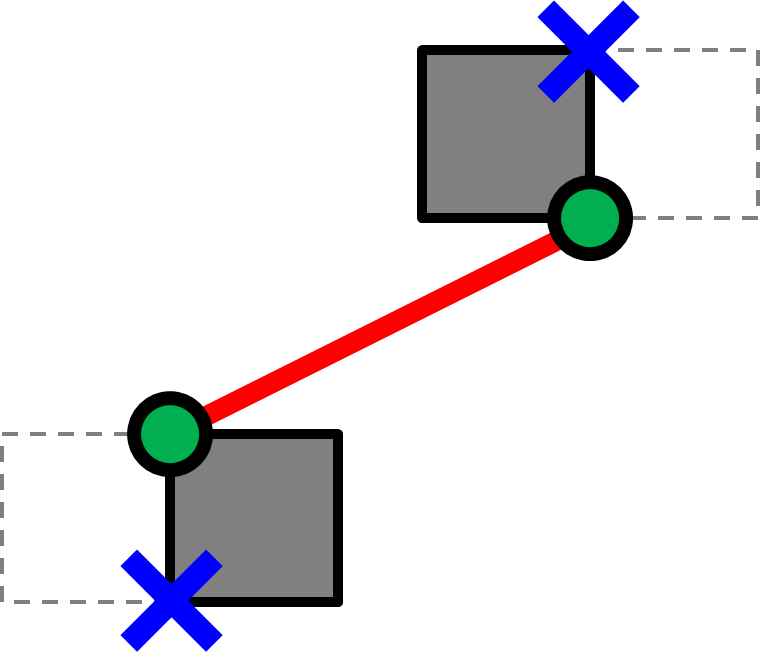}
      \caption{}
      \label{fig:necessary_c2}
    \end{subfigure}
    \begin{subfigure}[b]{.35\linewidth}
      \centering
      \includegraphics[width=.9\linewidth]{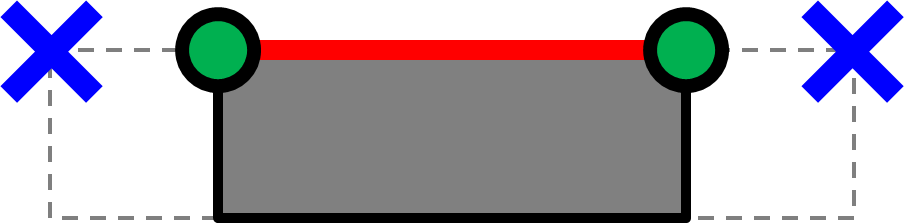}
      \caption{}
      \label{fig:necessary_c3}
    \end{subfigure}~~~~~%
    \begin{subfigure}[b]{.35\linewidth}
      \centering
      \includegraphics[width=.9\linewidth]{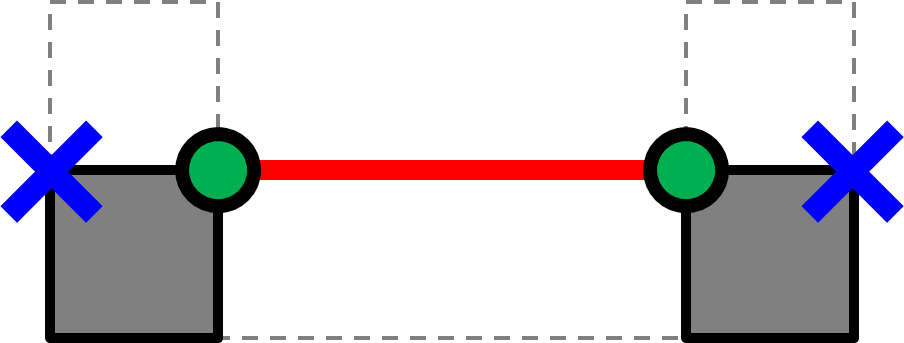}
      \caption{}
      \label{fig:necessary_c4}
    \end{subfigure}%
  \caption{The four possible types of edge in SVGs.}
  \label{fig:necessary_edges}
\end{figure}

We note that Theorem \ref{thm:necessary_edges} does not guarantee that an edge is necessary in the case of multiple optimal paths between the two points. However, this is uncommon, so pruning these edges will yield only a marginal running time improvement.

Experimental evaluation of SVGs can be found in the Experiments section at the end of the paper.

\section{Edge N-Level Sparse Visibility Graphs}

Before we discuss ENLSVGs, it is important to understand the flaws of the SVG Algorithm. Even though Theorem \ref{thm:necessary_edges} states that every edge in an SVG is necessary with a few rare exceptions, a large percentage of the edges are only useful for a small set of start-goal pairs. An example is Figure \ref{fig:edgeuvlimiteduse}, where edge $uv$ is only useful for constructing a path between the two marked points. Figure \ref{fig:limiteduseedgeclique} shows a clique of edges, each of which are useful for only a few start or goal points.

\begin{figure}[!h]
  \centering
    \begin{subfigure}[b]{.5\linewidth}
      \centering
      \includegraphics[width=.9\linewidth]{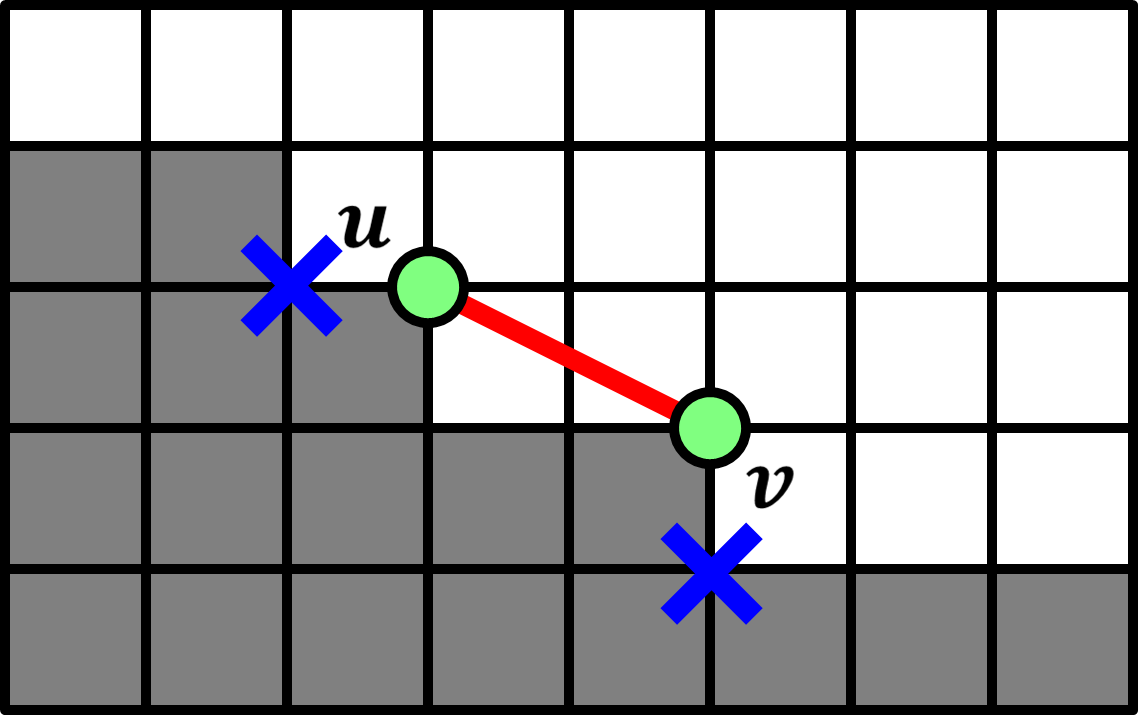}
      \caption{}
      \label{fig:edgeuvlimiteduse}
    \end{subfigure}%
    \begin{subfigure}[b]{.5\linewidth}
      \centering
      \includegraphics[width=.9\linewidth]{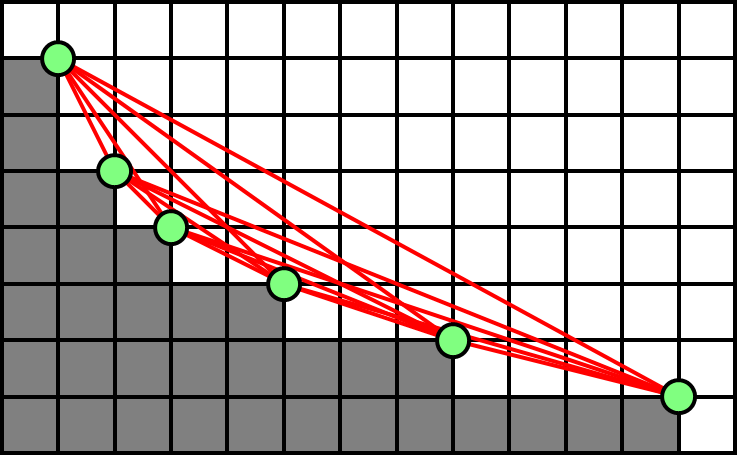}
      \caption{}
      \label{fig:limiteduseedgeclique}
    \end{subfigure}%
  \caption{Edges with limited usefulness within a concave section of blocked tiles.}
  \label{fig:limiteduseedges}
\end{figure}

\subsection{Edge Levels}

In SVGs, we prune outgoing edges that, when traversed, cannot be taut on the next hop. We extend this concept by looking further ahead than one hop, and prune outgoing edges that cannot be taut in future hops. We note that in SVGs, edges pruned due to a lack of taut exits are present in an optimal path only if they are the first or last hop of the search. We thus define the following concept of edge levels:

\begin{defn} Edge Level\\
\label{dfn:edgelevel}
An edge is level $k \geq 0$ if at any one of its endpoints, it has no taut neighbouring edge of level more than $k-1$, and if $k > 0$, also has a taut neighbouring edge of level $k-1$. Edges that do not fit this definition have level $\infty$.
\end{defn}

All ``edges" not in the Sparse Visibility Graph are referred to as Level-0 edges. The idea is that for an edge $e$ of level $\ell$, for any taut path that passes through $e$, edge $e$ must be the $k$\textsuperscript{th} hop from one of the endpoints of the path, for some $k \leq \ell$. If we were to restrict our search to only edges of increasing level from either end, all taut paths will be considered in the search, maintaining optimality.

\begin{algorithm}
\caption{ComputeEdgeLevels}\label{alg:computeEdgeLevels}
\begin{algorithmic}[1]
\Procedure{ComputeEdgeLevels}{$E$}
    \ForEach {$e = (u,v) \in E$}
\State      $e$.level $\gets \infty$
    \EndFor
\State  hasChanges $\gets$ True
\State  $\ell \gets 1$
    \While {hasChanges}
\State  hasChanges $\gets$ False
        \ForEach {$e = (u,v) \in E$}
            \If {$u$ or $v$ has no taut exit of level $\geq \ell$}
\State              $e$.level $\gets \ell$
\State              hasChanges $\gets$ True
            \EndIf
        \EndFor
\State  $\ell \gets \ell + 1$
    \EndWhile
\EndProcedure
\end{algorithmic}
\end{algorithm}

These edge levels can be computed by iteratively pruning edges level-by-level as shown in Algorithm \ref{alg:computeEdgeLevels}. Note that the computed edge levels are independent of the order the edges are selected in line 8. A simple algorithm is used for ease of understanding, though we believe that this procedure can be implemented with a more efficient algorithm.

Figures \ref{fig:edgelevels}, \ref{fig:edgelevelszoom} illustrate an example of edge levels on an SVG. Green edges are level $1$, yellow edges are level $2$, orange edges are level $3$, and the single red edge is level $4$.

\begin{figure}[!h]
  \centering
    \includegraphics[width=.75\linewidth]{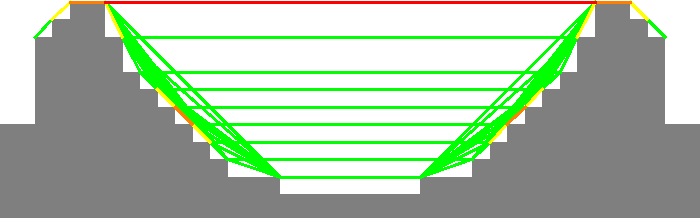}
  \caption{Edge levels in a Sparse Visibility Graph}
  \label{fig:edgelevels}
\end{figure}

\begin{figure}[!h]
  \centering
    \includegraphics[width=.55\linewidth]{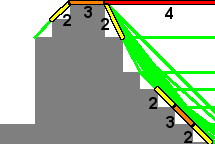}
  \caption{Zoom-in of Figure \ref{fig:edgelevels}, with edges of level $2$ and above outlined and labelled.}
  \label{fig:edgelevelszoom}
\end{figure}

With the edge levels defined as before, we have the following results on the edge levels:\\

\begin{lem}
\label{lem:tautpathfinitelevels}
Consider any taut path. Let the levels of the edges along be the path be $\ell_1,\ell_2,\cdots,\ell_n$ respectively.
Then for each $i \in \{2,3,\cdots,n-1\}$, if $\ell_i$ is finite, then either $\ell_{i-1} < \ell_i$ or $\ell_{i+1} < \ell_i$.
\end{lem}
\begin{proof}
Suppose that there is an $i$ such that $\ell_{i-1} \geq \ell_i$ and $\ell_{i+1} \geq \ell_i$. As the subpaths $e_{i-1}e_i$ and $e_ie_{i+1}$ are taut, this means edge $e_i$ has neighbouring edges on both endpoints with level $\geq \ell_i$, implying the level of $e_i$ is at least $\ell_i+1$, which is a contradiction.
\end{proof}

\begin{thm}
\label{thm:tautpathfinitelevels}

Assuming that every edge has a finite level, the sequence of edges of any taut path between the Start and the Goal vertices will be of the form

\begin{center}
$e_1e_2\cdots e_k e'_{k+1} \cdots e'_n$
\end{center}

where edges $e_1e_2\cdots e_k$ have strictly increasing levels, and $e'_{k+1}e'_{k+2}\cdots e_n$ have strictly decreasing levels.\\
\end{thm}

\begin{proof}
Let the levels of the edges along be the path be $\ell_1,\ell_2,\cdots,\ell_n$ respectively. From the lemma, we can see that in the path, if there is an $i$ such that $\ell_i \geq \ell_{i+1}$, then we must have $\ell_{i+1} > \ell_{i+2}$, implying $\ell_{i+2} > \ell_{i+3}$ and so on, inductively proving that the remaining edges of the path will have strictly decreasing levels, proving Theorem \ref{thm:tautpathfinitelevels}.
\end{proof}

\subsection{Level-W Edges}

Not all edges will be assigned a finite level by Algorithm \ref{alg:computeEdgeLevels}. We refer to the remaining edges (with level $\infty$) as level-W edges. Notably, an edge is level-W if and only if it is part of some taut cycle. Examples of taut cycles are shown in Figure \ref{fig:tautcycles}. We observe that the graph of level-W edges is similar to the tangent graphs described in \cite{tangent_graphs}, as taut cycles wrap around the convex hulls of obstacles.

\begin{figure}[!h]
  \centering
    \begin{subfigure}[b]{.37\linewidth}
      \centering
      \includegraphics[width=\linewidth]{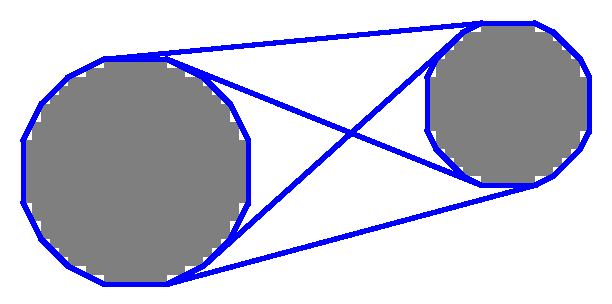}
      \caption{}
      \label{fig:tautcycle2}
    \end{subfigure}%
    \begin{subfigure}[b]{.37\linewidth}
      \centering
      \includegraphics[width=\linewidth]{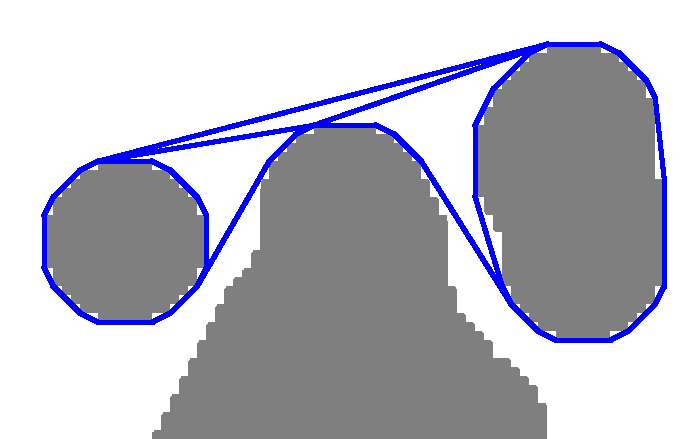}
      \caption{}
      \label{fig:tautcycle3}
    \end{subfigure}%
    \begin{subfigure}[b]{.26\linewidth}
      \centering
      \includegraphics[width=\linewidth]{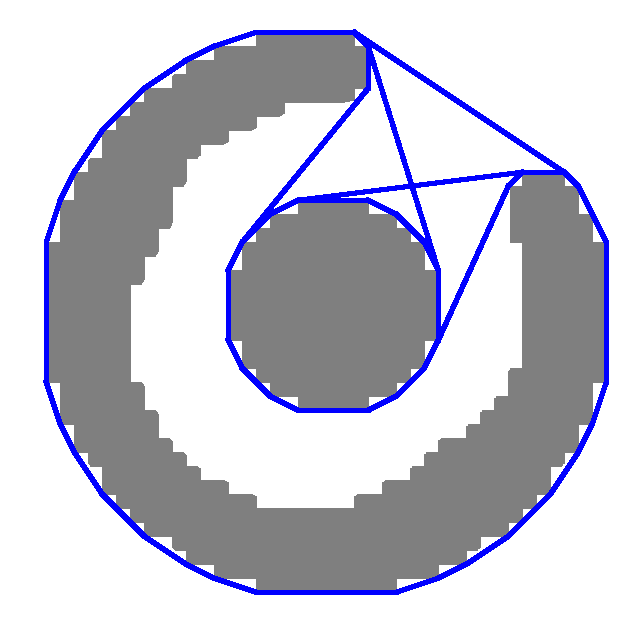}
      \caption{}
      \label{fig:tautcycle4}
    \end{subfigure}%
  \caption{Examples of cycles of taut edges.}
  \label{fig:tautcycles}
\end{figure}

With this definition of level-W edges, we have a similar theorem regarding the edge levels of taut paths.\\


\begin{thm}
\label{thm:tautpathlevels}

The sequence of edges of any taut path between the start and goal vertices will be of the form:

\begin{center}
$e_1e_2\cdots e_{k_1} w_{k_1+1} w_{k_1+2} \cdots w_{k_2} e'_{k_2+1} e'_{k_2+2} \cdots e'_n$
\end{center}

where edges $e_1e_2\cdots e_{k_1}$ have strictly increasing levels, $w_{k_1+1} w_{k_1+2} \cdots w_{k_2}$ are level-W, and $e'_{k_2+1}e'_{k_2+2}\cdots e_n$ have strictly decreasing levels.
\end{thm}
\begin{proof}
The proof is similar to that of Theorem \ref{thm:tautpathfinitelevels}. Let the levels of the edges along be the path be $\ell_1,\ell_2,\cdots,\ell_n$ respectively. From Lemma \ref{lem:tautpathfinitelevels}, we can see that in the path, if there is an $i$ such that $\ell_i \geq \ell_{i+1}$ and $\ell_{i+1}$ is finite, then we must have $\ell_{i+1} > \ell_{i+2}$, implying $\ell_{i+2} > \ell_{i+3}$ and so on, inductively proving that the remaining edges of the path will have strictly decreasing levels. This gives the required form.
\end{proof}

As the optimal path is taut, it obeys the rule in Theorem \ref{thm:tautpathlevels}. Thus, other than near the start and the end of the route, we only have to search level-W edges to ensure that the optimal path is included in the search.

\subsection{Edge Marking and Search}

Edge Marking is how we make use of Theorem \ref{thm:tautpathlevels} in the search. Before the search, from both the start and goal vertices, we run a depth-first search to mark all finite-level edges reachable by a taut path of strictly increasing edge levels. We stop the search when we reach Level-W edges. Figure \ref{fig:enlsvg_marking} illustrates the edges that are marked this way.

\begin{figure}[!h]
  \centering
    \begin{subfigure}[b]{.5\linewidth}
      \centering
      \includegraphics[width=.9\linewidth]{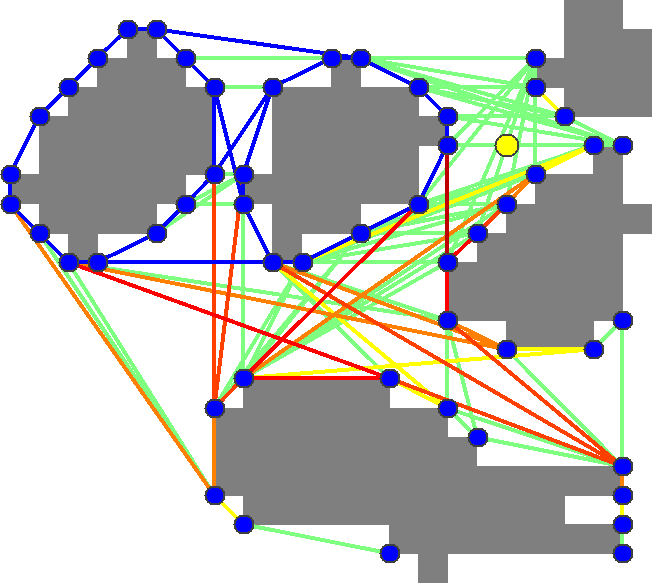}
      \caption{Before Marking}
      \label{fig:enlsvg_beforemark}
    \end{subfigure}%
    \begin{subfigure}[b]{.5\linewidth}
      \centering
      \includegraphics[width=.9\linewidth]{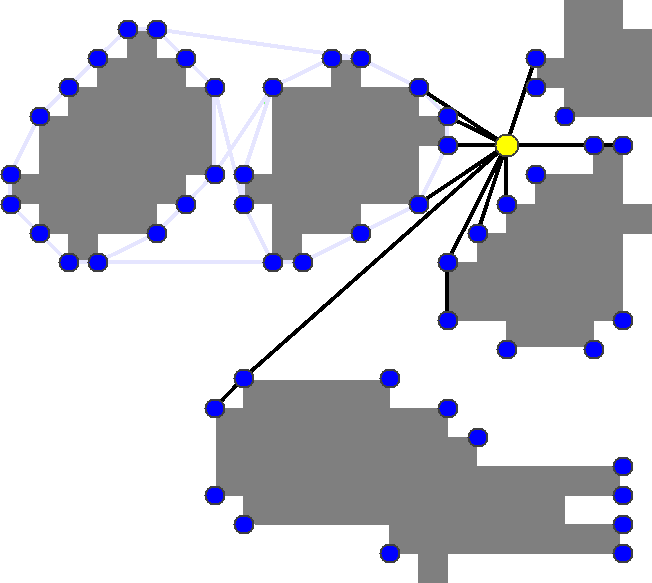}
      \caption{After Marking}
      \label{fig:enlsvg_aftermark}
    \end{subfigure}%
  \caption{Illustration of the edge marking process.}
  \label{fig:enlsvg_marking}
\end{figure}

We then restrict our A* search to use only marked edges and level-W edges. Let $H$ denote the subgraph induced by the marked and level-W edges. To prove that the algorithm is optimal, it suffices to show that the optimal path resides within the graph $H$.

\begin{thm}
All taut paths (including the optimal path) from the start $s$ to the goal $t$ reside within the graph $H$.
\end{thm}

\begin{proof}
Consider any taut path from $s$ to $t$. It must have edges of the form in Theorem \ref{thm:tautpathlevels}. Edges $e_1,e_2\cdots ,e_{k_1}$ would have been marked from $s$ as each of these edges are reachable by a taut path of increasing edge levels from $s$. Similarly, edges $e'_{k_2+1},\cdots e'_n$ will be marked from $t$. Thus all of the edges in the path are either marked or Level-W, and so are in $H$.
\end{proof}

\noindent
Thus the algorithm can be summarised in three steps:
\begin{enumerate}
\item
Insert start, goal into the graph using Line-of-Sight Scans.
\item
Mark reachable edges from the start and goal vertices.
\item
Compute optimal path to the goal by running Taut A* on only the marked and Level-W edges.
\end{enumerate}

\subsection{Skip-Edges}

The graph of level-W edges can be further reduced through the concept of Skip-Edges. As seen in Figure \ref{fig:largeconvexhulls_w}, large convex hulls can produce long, unbranching paths of level-W edges. Each unbranching path can be reduced to a single edge with weight equal to the length of the path as shown in Figure \ref{fig:largeconvexhulls_skip}. We refer to these edges as Skip-Edges. Skip-Edges makes the search time dependent on the amount of detail in the map, rather than the scale of the map.

\begin{figure}[!h]
  \centering
    \begin{subfigure}[b]{.5\linewidth}
      \centering
      \includegraphics[width=.8\linewidth]{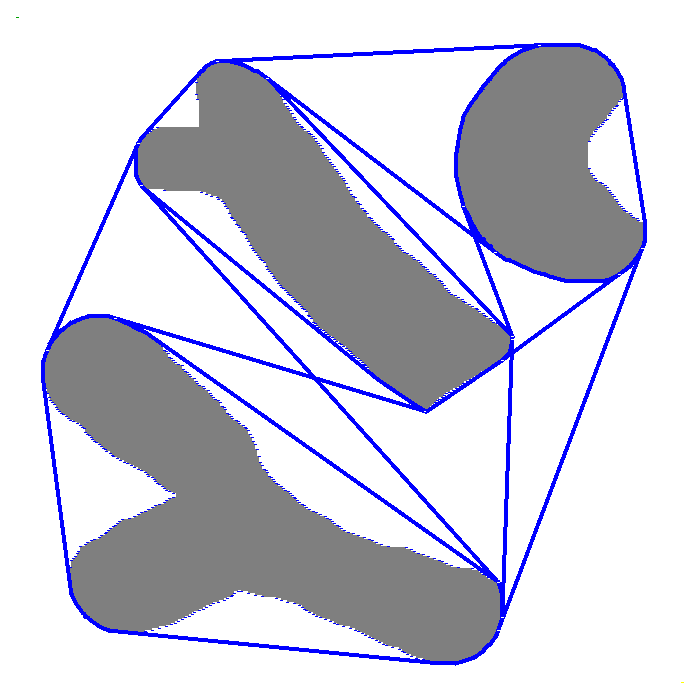}
      \caption{level-W Edges}
      \label{fig:largeconvexhulls_w}
    \end{subfigure}%
    \begin{subfigure}[b]{.5\linewidth}
      \centering
      \includegraphics[width=.8\linewidth]{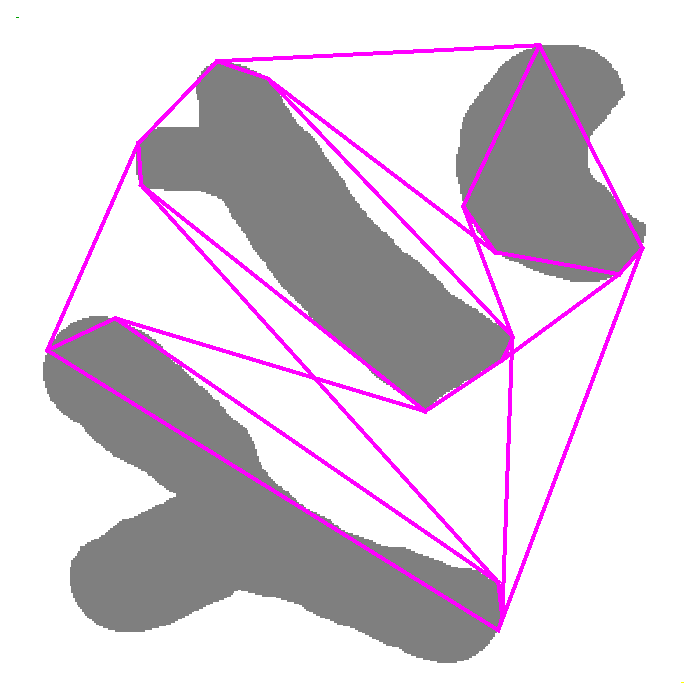}
      \caption{Skip-Edges}
      \label{fig:largeconvexhulls_skip}
    \end{subfigure}
  \caption{Skip-Edge network derived from level-W edges.}
  \label{fig:largeconvexhulls}
\end{figure}

To construct the Skip-Edge network, consider the graph $W$ induced by the level-W edges. All vertices of degree at least $3$ in $W$ are identified as Skip-Vertices. We then trace the unbranching paths of Level-W edges between Skip-Vertices to form the Skip-Edge network.

We also make a slight change to the marking scheme. If we reach a Level-W edge while marking edges of increasing level, we continue marking subsequent Level-W edges until a Skip-Vertex is reached. This is illustrated in Figure \ref{fig:enlsvg_markingwithskip}. We then run Taut A* on only marked edges and Skip-Edges.

\begin{figure}[!h]
  \centering
    \begin{subfigure}[b]{.5\linewidth}
      \centering
      \includegraphics[width=.9\linewidth]{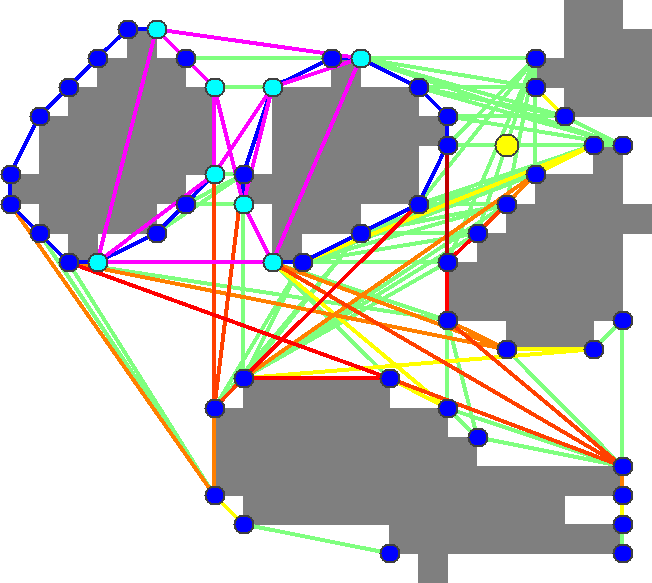}
      \caption{Before Marking}
      \label{fig:enlsvg_beforemarkwithskip}
    \end{subfigure}%
    \begin{subfigure}[b]{.5\linewidth}
      \centering
      \includegraphics[width=.9\linewidth]{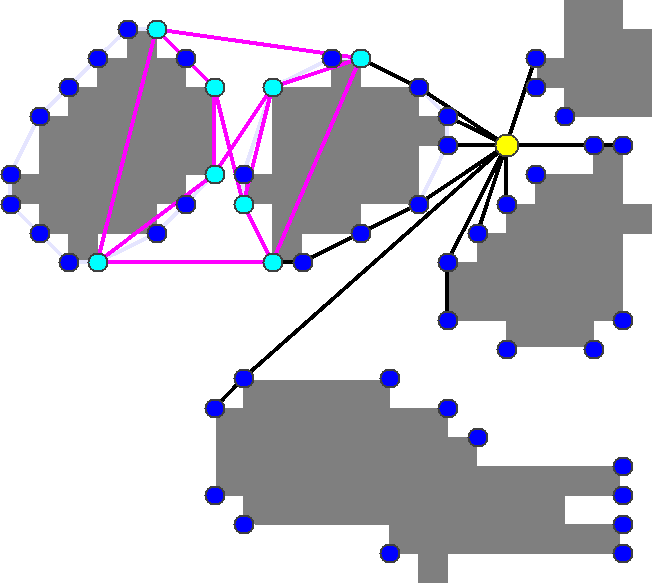}
      \caption{After Marking}
      \label{fig:enlsvg_aftermarkwithskip}
    \end{subfigure}%
  \caption{The edge marking process with Skip-Edges.}
  \label{fig:enlsvg_markingwithskip}
\end{figure}

\subsection{Search Tree Comparison}
Figure \ref{fig:ebonlakes_searchtrees2} illustrates the difference between the search trees of the SVG and ENLSVG algorithms. We can observe the running time improvement through how much of the original search tree the algorithms prune.

\begin{figure}[!h]
  \centering
    \begin{subfigure}[b]{.5\linewidth}
      \centering
      \includegraphics[width=.9\linewidth]{diagrams/ebonlakes_SVG.png}
      \caption{SVG}
      \label{fig:ebonlakes_svg2}
    \end{subfigure}%
    \begin{subfigure}[b]{.5\linewidth}
      \centering
      \includegraphics[width=.9\linewidth]{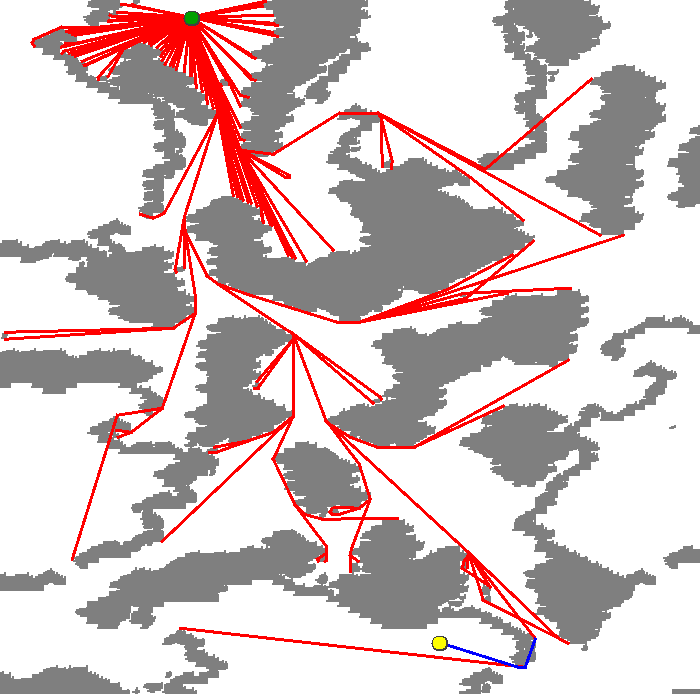}
      \caption{ENLSVG}
      \label{fig:ebonlakes_enlsvg}
    \end{subfigure}%
  \caption{Search tree comparson on the map EbonLakes.}
  \label{fig:ebonlakes_searchtrees2}
\end{figure}

\section{Experimental Results}

\subsection{Algorithm Running Time}
The algorithms compared are Theta*, the original Visibility Graph algorithm using Line-of-Sight Checks (VG\textsubscript{C}), Visibility Graphs using Line-of-Sight Scans (VG\textsubscript{S}), SVGs and ENLSVGs. Running times on each map are averaged over $1000$ to $2500$ runs using randomly-picked pairs of reachable points. The algorithms were implemented in Java\footnote{The implementations are available at \url{github.com/Ohohcakester/Any-Angle-Pathfinding}} on a 2.60 GHz Intel i5 processor with 8GB RAM.

\newcolumntype{?}{!{\vrule width 1pt}}
\begin{table*}[t]
\centering
\begin{adjustbox}{width=\linewidth}
\begin{tabular}{?c?r|r|r|r|r?r|r|r?r|r?}
\hline        Type &   Theta* &   VG\textsubscript{C} &   VG\textsubscript{S}  &   SVG  &ENLSVG  &Insert\textsubscript{ENLSVG}  &Mark\textsubscript{ENLSVG} &Search\textsubscript{ENLSVG} & Insert\textsubscript{SVG} &Search\textsubscript{SVG}\\
\hline  gen30\_2000 &$  225.64$&$ 55.57$&$  7.24 $&$ 3.18 $&$ 2.62 $&$        1.32 $&$         1.06 $&$       0.24 $&$   1.24 $&$   1.94 $\\
\hline  gen30\_4000 &$  999.52$&$222.61$&$ 24.66 $&$ 9.61 $&$ 4.38 $&$        1.86 $&$         1.58 $&$       0.95 $&$   1.78 $&$   7.83 $\\
\hline  gen30\_6000 &$ 2589.36$&$532.87$&$ 55.82 $&$22.83 $&$ 6.41 $&$        2.07 $&$         1.93 $&$       2.40 $&$   1.97 $&$  20.85 $\\
\hline  gen45\_2000 &$  383.09$&$ 45.88$&$ 11.29 $&$ 5.11 $&$ 1.34 $&$        0.63 $&$         0.47 $&$       0.24 $&$   0.66 $&$   4.45 $\\
\hline  gen45\_4000 &$ 1921.94$&$188.90$&$ 42.74 $&$18.80 $&$ 2.47 $&$        0.90 $&$         0.67 $&$       0.90 $&$   0.89 $&$  17.90 $\\
\hline  gen45\_6000 &$ 6124.34$&$452.53$&$116.27 $&$47.41 $&$ 3.69 $&$        0.96 $&$         0.76 $&$       1.97 $&$   0.93 $&$  46.47 $\\
\hline scaled\_2048 &$  570.85$&$ 19.66$&$  4.08 $&$ 2.01 $&$ 1.90 $&$        1.32 $&$         0.52 $&$       0.06 $&$   1.25 $&$   0.75 $\\
\hline scaled\_4096 &$ 4173.37$&$106.95$&$ 11.92 $&$ 5.60 $&$ 5.35 $&$        3.13 $&$         2.07 $&$       0.14 $&$   2.96 $&$   2.64 $\\
\hline scaled\_6144 &$15904.88$&$263.17$&$ 24.25 $&$10.01 $&$ 9.82 $&$        5.20 $&$         4.41 $&$       0.21 $&$   4.86 $&$   5.15 $\\
\hline  tiled\_2048 &$  357.94$&$ 30.84$&$  7.00 $&$ 3.53 $&$ 0.78 $&$        0.34 $&$         0.17 $&$       0.26 $&$   0.36 $&$   3.17 $\\
\hline  tiled\_4096 &$ 1752.23$&$131.37$&$ 30.61 $&$13.72 $&$ 1.84 $&$        0.48 $&$         0.23 $&$       1.14 $&$   0.45 $&$  13.27 $\\
\hline  tiled\_6144 &$ 4798.91$&$296.96$&$ 80.08 $&$37.38 $&$ 4.00 $&$        0.49 $&$         0.27 $&$       3.24 $&$   0.48 $&$  36.90 $\\
\hline 
\end{tabular}
\end{adjustbox}
\caption{Running times (in milliseconds) for generated, scaled and tiled maps of sizes $2000\times 2000$ to $6000\times 6000$. \\ The left half of the table compares the different algorithms. The right half breaks down the ENLSVG and SVG algorithms.}
\label{tab:large}
\end{table*}

\begin{table}[t]
\centering
\begin{adjustbox}{width=.9\columnwidth}
\begin{tabular}{|c|r|r|r|r|r|}
\hline     Maps& Theta* &  VG\textsubscript{C}  & VG\textsubscript{S}  &   SVG  &ENLSVG\\
\hline      wc3&$  5.68$&$ 0.77 $&$0.36 $&$ 0.30 $&$ 0.35 $\\
\hline    bg512&$  5.43$&$ 0.43 $&$0.30 $&$ 0.27 $&$ 0.30 $\\
\hline      sc1&$ 40.59$&$ 5.43 $&$1.76 $&$ 0.87 $&$ 0.57 $\\
\hline random10&$  4.16$&$23.81 $&$8.21 $&$ 1.46 $&$ 1.70 $\\
\hline random20&$  6.37$&$18.37 $&$8.86 $&$ 3.05 $&$ 3.12 $\\
\hline random30&$  8.94$&$16.71 $&$8.75 $&$ 4.31 $&$ 4.26 $\\
\hline random40&$  9.49$&$10.20 $&$6.22 $&$ 3.52 $&$ 2.94 $\\
\hline 
\end{tabular}
\end{adjustbox}
\caption{Running times (ms) on benchmark maps.}
\label{tab:benchmark}
\end{table}

\begin{table}[t]
\centering
\begin{adjustbox}{width=.8\columnwidth}
\begin{tabular}{|c|r|r|r|r|}
\hline     Maps&   VG\textsubscript{C}  & VG\textsubscript{S}  &   SVG  &ENLSVG\\
\hline      wc3&$   170 $&$ 143 $&$   93 $&$  100 $\\
\hline    bg512&$    83 $&$  67 $&$   53 $&$   55 $\\
\hline      sc1&$  3220 $&$ 470 $&$  291 $&$  388 $\\
\hline random10&$194026 $&$7838 $&$ 2877 $&$ 8305 $\\
\hline random20&$268934 $&$2928 $&$ 1096 $&$ 3585 $\\
\hline random30&$229677 $&$1317 $&$  539 $&$ 1670 $\\
\hline random40&$ 63610 $&$ 395 $&$  194 $&$  572 $\\
\hline 
\end{tabular}
\end{adjustbox}
\caption{Graph construction times (ms) on benchmark maps.}
\label{tab:construction}
\end{table}

\begin{figure}[!h]
  \centering
    \begin{subfigure}[b]{.5\linewidth}
      \centering
      \includegraphics[width=.85\linewidth]{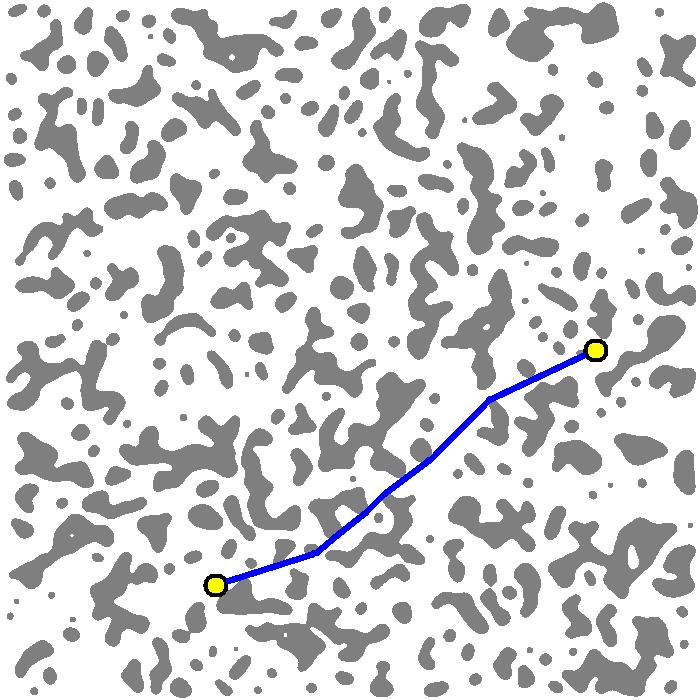}
      \caption{Generated, $30\%$ blocked}
      \label{fig:map_gen_30}
    \end{subfigure}%
    \begin{subfigure}[b]{.5\linewidth}
      \centering
      \includegraphics[width=.85\linewidth]{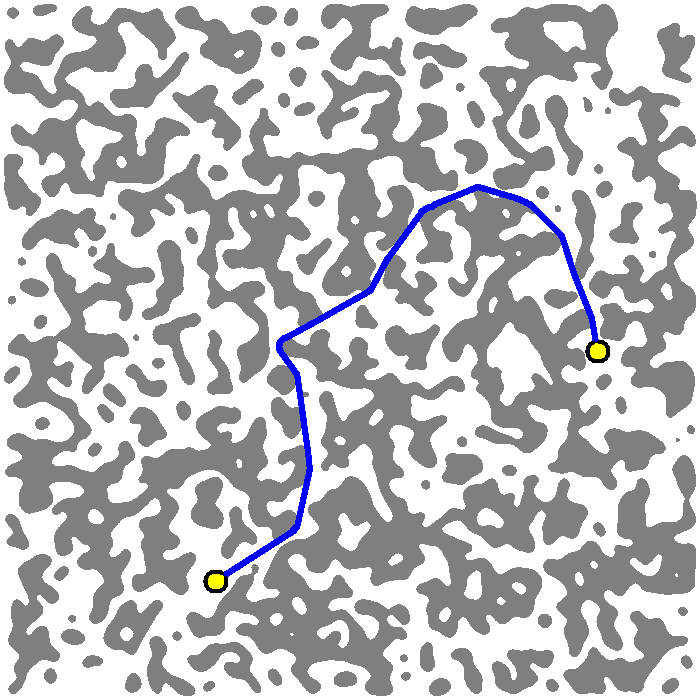}
      \caption{Generated, $45\%$ blocked}
      \label{fig:map_gen_45}
    \end{subfigure}
  \centering
    \begin{subfigure}[b]{.5\linewidth}
      \centering
      \includegraphics[width=.85\linewidth]{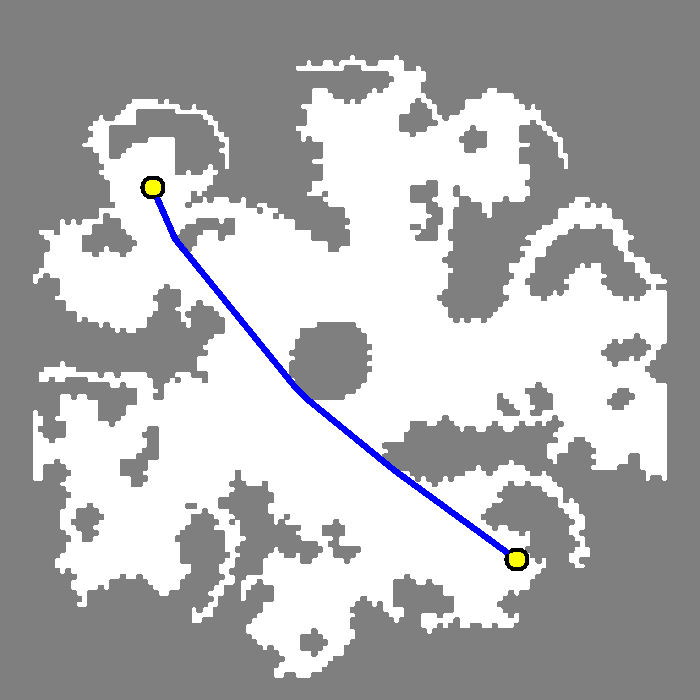}
      \caption{Upscaled benchmark map}
      \label{fig:map_scaled}
    \end{subfigure}%
    \begin{subfigure}[b]{.5\linewidth}
      \centering
      \includegraphics[width=.85\linewidth]{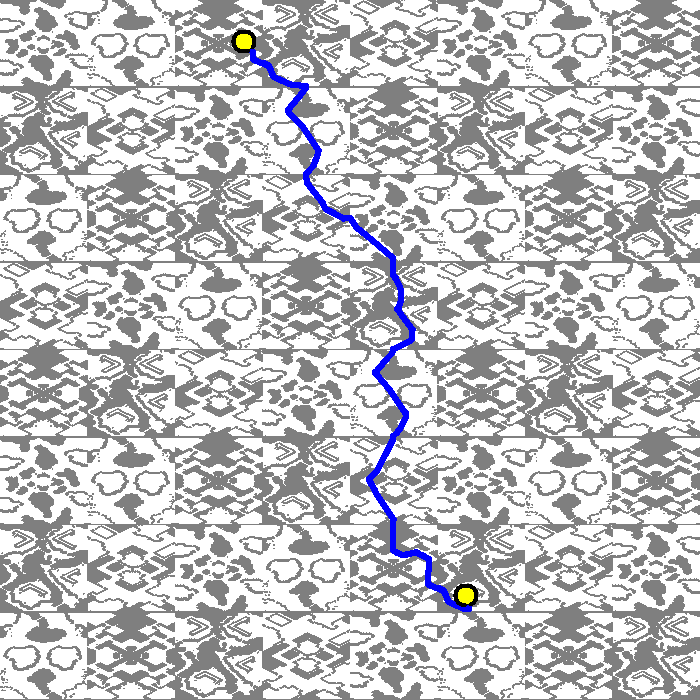}
      \caption{Tiled benchmark maps}
      \label{fig:map_tiled}
    \end{subfigure}%
  \caption{Some $4000\times 4000$ maps used in the experiments.}
  \label{fig:modified_benchmarks}
\end{figure}

In addition to the benchmarks from \cite{sturtevant2012benchmarks}, we use three methods to generate larger maps. The first two sets are random cave maps generated using cellular automata \cite{cellular_automata}, a common technique for generating game maps, with $30\%$ / $45\%$ blocked tiles respectively. The third set is upscaled versions of benchmark game maps, smoothed using cellular automata. The fourth set is generated by tiling benchmark game maps to form larger maps. Examples are shown in Figure \ref{fig:modified_benchmarks}. Map sizes used are around $2000\times 2000$, $4000\times 4000$, and $6000\times 6000$.

From Tables \ref{tab:large} and \ref{tab:benchmark}, we see that the biggest runtime saving from VGs to SVGs comes from using Line-of-Sight Scans (VG\textsubscript{S}) instead of Line-of-Sight Checks (VG\textsubscript{C}) for inserting start and goal points. SVGs improve the running time further with no additional cost. The improvement in construction time from using Line-of-Sight Scans is especially significant on random maps. SVG construction is even faster than VG\textsubscript{S} as it can use Taut-Direction Line-of-Sight Scans.

This is because the running time of Line-of-Sight Checks depends on the number of visibility graph vertices (convex corners) in the entire grid, while the running time of Line-of-Sight Scans depends on the number of visible vertices from each vertex. Random maps have more convex corners and shorter visibility ranges.

While SVGs effectively cut running time regardless of map structure, ENLSVGs speed up search by taking advantage of map structure to build a hierarchy. As such, the cost savings are small on completely random maps (Table \ref{tab:benchmark}).

We break down the ENLSVG algorithm into three components: insertion, marking and search. Insertion is the Line-of-Sight Scans to connect the start and goal to the graph. Search refers to the final A* search. Only the insertion and search components apply to the SVG algorithm.

From Table \ref{tab:large}, we see ENLSVGs perform a lot better than SVGs on tiled and generated maps, but only slightly better on upscaled maps. This is because the bottleneck on upscaled maps is the insertion and marking steps, while the bottleneck on tiled maps is the search step. Insertion time depends on how wide the open spaces are, while search time is tied to the complexity of the map. Upscaled maps have large open spaces and low complexity, while tiled maps are the opposite. If we look at search time alone however, we see that ENLSVGs consistently do much better than SVGs.

When compared to existing algorithms Theta* and VG\textsubscript{C}, especially on large maps, SVGs and ENLSVGs are orders of magnitude faster. Theta* was used as a benchmark as the running time of Theta* is well understood, as compared to algorithms like Anya. The relationship between Theta* and other algorithms can be found in \cite{comparison_paper}. ENLSVGs perform well (in real-time) even on $10000\times 10000$ grids. Memory constraints from pathfinding on large grids however prevent us from extracting reliable running time data, due to inconsistent running times when memory paging occurs.

The difference in construction time between ENLSVGs and SVGs is the time spent building the hierarchy over the SVG. As construction time is not the main consideration in this paper, the hierarchy building step is implemented using Algorithm \ref{alg:computeEdgeLevels} above, which runs in $O(m^2 \Delta)$ time on an SVG with $m$ edges and maximum vertex degree $\Delta$. We believe this can be improved to $O(m\Delta)$ with a better algorithm.

\section{Conclusions}
On maps with wider open spaces, even though the ENLSVG algorithm's running time is bottlenecked by the insertion and marking steps, the ultimate reduction in time used for the search step gives is an indication of the potential speedup that can be obtained through the use of ENLSVGs.

If we could do away with the insertion and marking steps, this speedup could be achieved. For example, pre-processing could be used to quickly find the visible neighbours of the start and goal points. Regarding the marking step, we can see that some form of goal-based initial search is needed in this algorithm, in order for the search to ``know'' that it is approaching the goal, and start following paths of decreasing edge levels. Whether a double-ended search algorithm can be used to omit the marking step is an open question.

In this paper, we also see the use of taut (locally optimal) paths as a heuristic for globally optimal paths to reduce the search space. The exact relationship between ENLSVGs (edge levels by pruning non-taut paths) and the 8-directional optimal algorithm N-Level Subgoal Graphs (vertex levels by pruning suboptimal paths), as well as a deeper analysis of their running times and pitfalls, remains to be investigated.

\section{Additional Experiments}
Two other recent Any-Angle Pathfinding algorithms have also been shown to perform significantly better than Theta* in practice. The first algorithm is Anya as described in \cite{harabor_anya_full}, a fast online optimal algorithm based on searching intervals instead of vertices. The second algorithm, in \cite{vg_quadtrees}, describes multiple optimisations to speed up A* on a Visibility Graph over a quadtree. An experimental comparison between ENLSVGs and these algorithms is in the works.

\bibliographystyle{aaai}
\bibliography{report}

\end{document}